\documentclass[pdflatex,sn-mathphys-num]{sn-jnl}

\usepackage{graphicx}
\usepackage{multirow}
\usepackage{amsmath,amssymb,amsfonts}
\usepackage{amsthm}
\usepackage[title]{appendix}
\usepackage{xcolor}
\usepackage{textcomp}
\usepackage{manyfoot}
\usepackage{booktabs}
\usepackage{tikz}
\usepackage{algorithm}
\usepackage{algpseudocode}
\usetikzlibrary{arrows.meta,positioning,shapes.geometric}
\hypersetup{pdftitle={Dimension Rigidity and Projective Geometry of Trace-Product Switchings of the Gold Cube},pdfauthor={Oleksandr Kuznetsov},pdfkeywords={APN functions, Gold function, switching, relative trace, finite fields, projective geometry}}

\theoremstyle{thmstyleone}
\newtheorem{theorem}{Theorem}[section]
\newtheorem{proposition}[theorem]{Proposition}
\newtheorem{lemma}[theorem]{Lemma}
\newtheorem{corollary}[theorem]{Corollary}
\theoremstyle{thmstyletwo}
\newtheorem{remark}[theorem]{Remark}

\theoremstyle{thmstylethree}

\newcommand{\F}{\mathbb F}

\newcommand{\Tr}{\operatorname{Tr}}
\newcommand{\tr}{\operatorname{tr}}
\newcommand{\im}{\operatorname{im}}

\newcommand{\ord}{\operatorname{ord}}

\newcommand{\cB}{\mathcal B}

\newcommand{\trchar}{\chi_E}
\newcommand{\abs}[1]{\lvert#1\rvert}
\newcommand{\set}[1]{\left\{#1\right\}}

\begin{document}

\title[Trace-product switchings of the Gold cube]{Dimension Rigidity and Projective Geometry of Trace-Product Switchings of the Gold Cube}

\author*[1,2]{\fnm{Oleksandr} \sur{Kuznetsov}}\email{oleksandr.kuznetsov@uniecampus.it}

\affil*[1]{\orgdiv{Department of Theoretical and Applied Sciences}, \orgname{eCampus University}, \orgaddress{\street{Via Isimbardi 10}, \postcode{22060}, \city{Novedrate}, \state{CO}, \country{Italy}}}
\affil[2]{\orgdiv{Department of Intelligent Software Systems and Technologies, School of Computer Science and Artificial Intelligence}, \orgname{V.N. Karazin Kharkiv National University}, \orgaddress{\street{4 Svobody Sq.}, \postcode{61022}, \city{Kharkiv}, \country{Ukraine}}}

\abstract{We completely classify a natural scalar trace-product switching of the Gold almost perfect nonlinear function $x\mapsto x^3$ in every even dimension. Nontrivial switchings occur only for $n=4,6,8$: the admissible coefficients are, respectively, the nonzero trace-zero elements, the six elements of multiplicative order nine, and $\F_4^*$. For every even $n\geq10$, no nonzero coefficient is admissible. The infinite range is excluded by additive-character estimates on a Fermat cubic, with exact finite bridges for $n=10,12$. The raw coefficient lists for $n=6,8$ appeared earlier in Arshad's dissertation; our contribution is their intrinsic description, a proof uniform in the dimension, and the resulting dimension-rigidity theorem. We also classify normalized rank-two extensions in dimension eight by $\mathbb P^1(\F_4)$. A binary trace selector accepts two coefficient values at each non-base projective point, and the eight accepted marked switchings form exactly two extended-affine, hence two CCZ, classes. A centre-independent low-rank derivative criterion reduces each rank-$r$ candidate to $2^r-1$ membership tests in precomputed forbidden sets. The global APN classes reached are known; the results describe their local organization around the Gold centre and rule out this switching mechanism in all larger even dimensions.}


\maketitle

\section{Introduction}\label{sec:intro}

Almost perfect nonlinear (APN) functions are the vectorial Boolean functions with optimal differential uniformity in characteristic two.  They play a central role in the theory of cryptographic mappings and are tightly connected with finite geometry, coding theory, and incidence structures; see, for example, \cite{Carlet2021}.  The Gold maps $x\mapsto x^{2^k+1}$ with $\gcd(k,n)=1$ are the basic quadratic APN power functions \cite{Gold1968}.  A persistent theme is to modify a known APN map locally or in a small number of output directions while preserving the APN property.

Switching constructions are not new.  Budaghyan, Carlet and Leander constructed switched cubes such as $x^3+\Tr(x^9)$ \cite{BudaghyanCarletLeander2009}, and Edel and Pott developed a broad rank-one switching framework \cite{EdelPott2009}.  More recently, modifications on affine subspaces of small codimension have been characterized systematically \cite{TaniguchiEtAl2025}.  In particular, the common-trace-factor family $F(x)+t(x)L(x)$ has a necessary and sufficient hyperplane-injectivity criterion.  The low-rank derivative formulation below recovers that criterion and extends its operator form to several arbitrary quadratic Boolean kernels; no priority is claimed for the rank-one or common-trace-factor cases.

The historical boundary is especially important for the present family.  Chapter~5 of Arshad's dissertation \cite{Arshad2018} studies $x^3+\Tr(x)L(x)$ and modifications on four cosets of a codimension-two subspace.  Its Examples~5.14 and~5.24 report, in field coordinates, the complete admissible sets
\[
 \set{0,\beta^7,\beta^{14},\beta^{28},\beta^{35},\beta^{49},\beta^{56}}
 \quad(n=6)
\]
and
\[
 \set{0,1,\beta^{85},\beta^{170}}
 \quad(n=8),
\]
respectively.  Thus the raw coefficient lists are prior results.  What was not provided there is an intrinsic description of those lists, a proof valid across all even dimensions, or the projective classification of synchronized rank-two extensions.  The later journal treatment \cite{TaniguchiEtAl2025} gives general existence criteria but not the dimension ladder or the projective selector proved here.

\subsection*{Our contributions}
\begin{enumerate}
\item We determine the admissible coefficients of the scalar trace-product family in every even dimension: a trace-zero hyperplane minus zero for $n=4$, the six elements of order nine for $n=6$, $\F_4^*$ for $n=8$, and the empty set for every even $n\geq10$.
\item We prove the uniform nonexistence result above dimension eight by additive-character sums on a Fermat cubic, with exact finite bridges only for $n=10,12$.
\item In dimension eight we classify all normalized rank-two combinations of two projective trace forms. The parameter space is $\mathbb P^1(\F_4)$, and a binary trace selector chooses exactly two coefficient pairs at each of the four non-base points.
\item The resulting eight marked switchings form two Frobenius cycles and exactly two EA classes; Yoshiara's theorem \cite{Yoshiara2012} then gives two CCZ classes.
\item As a reusable tool, we give a centre-independent fixed-point criterion for coefficient-rank-$r$ quadratic updates. It yields precomputable forbidden sets and reduces each candidate to $2^r-1$ incidence tests.
\end{enumerate}
The reached global classes are known. The point is instead to describe how they occur around the Gold centre and why the phenomenon is confined to dimensions four, six, and eight. The practical value is methodological: the nonexistence theorem removes an entire switching mechanism from searches in larger even dimensions, while the low-rank criterion replaces repeated APN tests by precomputed incidence queries and can be transported to other quadratic APN centres.

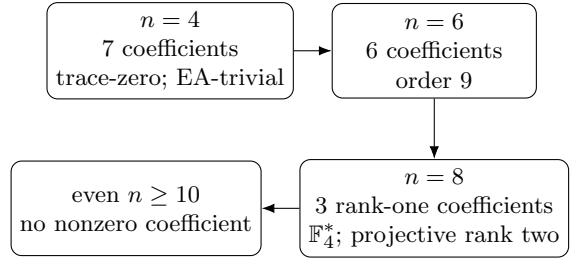
\begin{figure}[ht]
\centering
\begin{tikzpicture}[>=Latex, node distance=8mm and 5mm, every node/.style={font=\small}, box/.style={draw, rounded corners, align=center, minimum width=2.7cm, minimum height=1.25cm}]
\node[box] (n4) {\(n=4\)\\7 coefficients\\trace-zero; EA-trivial};
\node[box, right=of n4] (n6) {\(n=6\)\\6 coefficients\\order 9};
\node[box, below=of n6] (n8) {\(n=8\)\\3 rank-one coefficients\\\(\F_4^*\); projective rank two};
\node[box, left=of n8] (n10) {even \(n\ge10\)\\no nonzero coefficient};
\draw[->] (n4) -- (n6);
\draw[->] (n6) -- (n8);
\draw[->] (n8) -- (n10);
\end{tikzpicture}
\caption{Dimension rigidity of the scalar trace-product switching. The admissible set changes algebraic type in dimensions four, six and eight, and disappears in every larger even dimension.}
\label{fig:dimension-ladder}
\end{figure}

The paper contains computer-assisted finite steps, but their role is explicit.  The dimension-six and dimension-eight necessity arguments are compressed to 12 and 33 Frobenius-orbit witnesses.  The infinite nonexistence theorem is analytic for $n\geq14$, with exact finite bridges for $n=10,12$.  All scripts, exact JSON certificates, and hashes are supplied as Online Resource~1 and preserved in the Zenodo archival release.

\section{Preliminaries}\label{sec:prelim}

Let $V$ be an $n$-dimensional vector space over $\F_2$.  A map $F:V\to V$ is APN if, for every $a\ne0$ and every $b$, the equation
\[
 F(x+a)+F(x)=b
\]
has at most two solutions.  We normalize quadratic maps by $F(0)=0$ and write their polar form as
\[
 B_F(a,x)=F(a+x)+F(a)+F(x).
\]
For fixed $a$, the map $x\mapsto B_F(a,x)$ is linear.  A normalized quadratic map is APN if and only if
\begin{equation}\label{eq:polar-apn}
 \ker B_F(a,\cdot)=\langle a\rangle
 \qquad\text{for every }a\ne0.
\end{equation}
All Boolean quadratic forms below are normalized similarly.

Two maps are extended-affine (EA) equivalent if one is obtained from the other by affine permutations of the input and output together with an affine output term.  CCZ equivalence is affine equivalence of graphs.  EA equivalence implies CCZ equivalence, and for quadratic APN functions the converse holds by \cite{Yoshiara2012}.  We use spectra of ortho-derivatives as strongly discriminating EA invariants, following the computational invariant literature \cite{Kaleyski2021,BeierleLeanderPerrin2022}.

For even $n$, put
\[
 K=\F_{2^n},\qquad E=\F_4\subset K,
\]
and fix $\lambda\in E\setminus\F_2$, so $\lambda^2+\lambda+1=0$.  We use
\[
 T=\Tr_{K/E},\qquad \tr=\Tr_{E/\F_2}.
\]
Trace transitivity gives $\Tr_{K/\F_2}=\tr\circ T$.

\section{Low-rank derivative updates}\label{sec:lowrank}

Let $q=(q_1,\ldots,q_r):V\to\F_2^r$ be a vector of normalized quadratic Boolean forms and let $U:\F_2^r\to V$ be linear.  Define
\[
 G=F+U\circ q.
\]
For $a\in V$, set
\[
 R_a(x)=\bigl(B_{q_1}(a,x),\ldots,B_{q_r}(a,x)\bigr).
\]
Because every polar form is alternating, $R_a(a)=0$, so $R_a$ is well defined on $V/\langle a\rangle$.

\begin{theorem}[Low-rank fixed-point criterion]\label{thm:lowrank}
Let $F:V\to V$ be a normalized quadratic APN function.  Then $G=F+U\circ q$ is APN if and only if there are no $a\ne0$ and $y\ne0$ such that
\begin{equation}\label{eq:fixedpoint}
 U(y)\in\im B_F(a,\cdot)
 \quad\text{and}\quad
 R_a\!\left(\overline B_{F,a}^{-1}(U(y))\right)=y,
\end{equation}
where
\[
 \overline B_{F,a}:V/\langle a\rangle\longrightarrow\im B_F(a,\cdot)
\]
is the isomorphism induced by $B_F(a,\cdot)$.
\end{theorem}

\begin{proof}
The polar form of $G$ is
\[
 B_G(a,x)=B_F(a,x)+U(R_a(x)).
\]
Suppose $x\notin\langle a\rangle$ lies in its kernel and put $y=R_a(x)$.  If $y=0$, then $B_F(a,x)=0$, contradicting the APN property of $F$.  Hence $y\ne0$, $U(y)=B_F(a,x)$ belongs to the derivative image, and the class of $x$ modulo $\langle a\rangle$ is $\overline B_{F,a}^{-1}(U(y))$.  This gives \eqref{eq:fixedpoint}.  Conversely, a solution of \eqref{eq:fixedpoint} supplies a class different from zero in $V/\langle a\rangle$ and therefore an extra kernel vector of $B_G(a,\cdot)$.
\end{proof}

For a fixed centre and fixed kernels define, for each $y\ne0$, the forbidden set
\begin{equation}\label{eq:forbidden}
 \cB_y(F,q)=\set{B_F(a,x):a\ne0,\ R_a(x)=y}.
\end{equation}

\begin{corollary}[Forbidden-set test]\label{cor:forbidden}
The update $F+U\circ q$ is APN if and only if
\[
 U(y)\notin\cB_y(F,q)
 \qquad\text{for every }y\in\F_2^r\setminus\set0.
\]
After the sets $\cB_y$ have been precomputed, one candidate requires only $2^r-1$ membership tests.
\end{corollary}

\begin{proposition}[Specialization to the hyperplane-injectivity criterion]\label{prop:hequiv}
Let $t:V\to\F_2$ be nonzero and let $L:V\to V$ be binary linear. For $G(x)=F(x)+t(x)L(x)$, Theorem~\ref{thm:lowrank} is equivalent to the hyperplane-injectivity criterion of \cite{TaniguchiEtAl2025}.
\end{proposition}
\begin{proof}
The polar update is $t(a)L(x)+t(x)L(a)$. If $t(a)=1$, every class modulo $\langle a\rangle$ has a unique representative in $\ker t$, and the absence of an extra derivative-kernel vector is exactly the injectivity of $x\mapsto B_F(a,x)+L(x)$ on $\ker t$. If $t(a)=0$, a possible extra kernel vector must satisfy $t(x)=1$; by symmetry $B_F(a,x)=B_F(x,a)$, the same equation is the trace-one test with $x$ as derivative direction. Thus the all-direction fixed-point condition and the hyperplane criterion are equivalent.
\end{proof}
\begin{remark}
The novelty claimed here is the arbitrary-kernel, rank-$r$ operator formulation and its precomputed forbidden sets, not the hyperplane case.
\end{remark}

\begin{algorithm}[ht]
\caption{Testing a low-rank update after precomputation}\label{alg:forbidden}
\begin{algorithmic}[1]
\Require Quadratic APN centre $F$, kernels $q_1,\ldots,q_r$, candidate linear map $U$
\State Precompute $\cB_y(F,q)$ from \eqref{eq:forbidden} for every $0\ne y\in\F_2^r$
\For{$0\ne y\in\F_2^r$}
  \If{$U(y)\in\cB_y(F,q)$} \State \Return not APN \EndIf
\EndFor
\State \Return APN
\end{algorithmic}
\end{algorithm}
After the one-time construction of the forbidden sets, a candidate requires exactly $2^r-1$ membership tests. This is the operational advantage of the criterion in searches around other quadratic APN centres.

\section{The scalar trace-product family}\label{sec:family}

Define the Boolean quadratic form
\begin{equation}\label{eq:Qdef}
 Q_n(x)=\tr(T(x))\,\tr(\lambda T(x))
       =\Tr_{K/\F_2}(x)\,\Tr_{K/\F_2}(\lambda x)
\end{equation}
and, for $\theta\in K$, the quadratic map
\begin{equation}\label{eq:Ftheta}
 F_{n,\theta}(x)=x^3+\theta Q_n(x).
\end{equation}

\begin{lemma}[Polar trace form]\label{lem:polarQ}
The polar form of $Q_n$ is
\begin{equation}\label{eq:polarQ}
 B_Q(a,x)=\tr\bigl(T(a)T(x)^2\bigr).
\end{equation}
In particular, $B_Q(a,x)=1$ if and only if $T(a)$ and $T(x)$ are linearly independent over $\F_2$.
\end{lemma}

\begin{proof}
Expand the product in \eqref{eq:Qdef} at $a+x$ and cancel the pure terms.  For $u,v\in E$, a direct calculation using $\lambda^2+\lambda+1=0$ gives
\[
 \tr(u)\tr(\lambda v)+\tr(v)\tr(\lambda u)=\tr(uv^2).
\]
Taking $u=T(a)$ and $v=T(x)$ yields \eqref{eq:polarQ}.  The resulting alternating form on the two-dimensional $\F_2$-space $E$ is nondegenerate.
\end{proof}

\subsection{Relation with the previously computed examples}

Arshad's examples are written as
\begin{equation}\label{eq:arshadterm}
 x^3+\theta\Tr_{K/\F_2}(x)T(x).
\end{equation}
They represent the same quadratic polar update as \eqref{eq:Ftheta}.  Indeed, for $y\in E$ one has
\begin{equation}\label{eq:traceidentity}
 \tr(y)y=\tr(y)\tr(\lambda y)+\lambda^2\tr(y).
\end{equation}
The last term is $\F_2$-linear in $y$.  Hence \eqref{eq:arshadterm} differs from $F_{n,\theta}$ by a binary linear output map.  This identity allows us to interpret the field-coordinate sets in Examples~5.14 and~5.24 of \cite{Arshad2018} as the intrinsic sets described below.

\subsection{Section products}

For a two-dimensional $\F_2$-subspace $W=\langle a,x\rangle\subset K$, define
\[
 \pi(W)=\prod_{w\in W\setminus\set0}w=ax(a+x).
\]
The Gold polar form satisfies
\[
 B_{x^3}(a,x)=a^2x+ax^2=\pi(W).
\]

\begin{theorem}[Section-product criterion]\label{thm:section}
Let $n$ be even and $\theta\ne0$.  Put
\begin{equation}\label{eq:Sn}
 S_n=\set{c^3+d^3:T(c)=1,\ T(d)=0}.
\end{equation}
Then
\begin{equation}\label{eq:sectioncrit}
 F_{n,\theta}\text{ is APN}\quad\Longleftrightarrow\quad \theta\notin S_n.
\end{equation}
\end{theorem}

\begin{proof}
An extra kernel vector for the derivative in direction $a$ spans with $a$ a two-dimensional subspace $W$.  The kernel equation is
\[
 \pi(W)+\theta B_Q(a,x)=0.
\]
Since $\pi(W)\ne0$, this occurs exactly when $B_Q(a,x)=1$ and $\pi(W)=\theta$.  By Lemma~\ref{lem:polarQ}, the first condition says that $T|_W:W\to E$ is an isomorphism of two-dimensional $\F_2$-spaces.  Every $\F_2$-linear map $E\to K$ has a unique linearized representation $\phi(y)=cy+dy^2$, because the two Frobenius monomials $y$ and $y^2$ form a basis of $\operatorname{Hom}_{\F_2}(E,K)$ over $K$. The identity $T\circ\phi=\operatorname{id}_E$ is equivalent to
\[
 T(c)y+T(d)y^2=y\qquad(y\in E),
\]
which in turn is equivalent to $T(c)=1$ and $T(d)=0$. Finally, using $1+\lambda+\lambda^2=0$ and $\lambda^3=1$,
\begin{align*}
 \phi(1)\phi(\lambda)\phi(\lambda^2)
 &=(c+d)(c\lambda+d\lambda^2)(c\lambda^2+d\lambda)\\
 &=c^3+d^3.
\end{align*}
Thus the possible products of transverse sections are exactly $S_n$. Geometrically, $\pi(W)$ records the multiplicative label of a two-dimensional section $W$ transverse to the relative-trace fibres.
\end{proof}

\section{The exceptional dimensions}\label{sec:small}

\subsection{Dimension four}

\begin{proposition}\label{prop:n4}
For $n=4$ and $\theta\ne0$,
\[
 F_{4,\theta}\text{ is APN}
 \quad\Longleftrightarrow\quad
 \Tr_{\F_{16}/\F_2}(\theta)=0.
\]
Every admissible map is EA-equivalent to $x^3$.
\end{proposition}

\begin{proof}
In $\F_{16}$, direct polarization gives
\[
 Q_4(x)=\Tr_{\F_{16}/\F_2}(x^3)+\ell(x)
\]
for a binary linear form $\ell$.  Hence
\[
 F_{4,\theta}(x)=B_\theta(x^3)+\theta\ell(x),
 \qquad B_\theta(y)=y+\theta\Tr_{\F_{16}/\F_2}(y).
\]
The rank-one linear update $B_\theta$ is invertible if and only if $\Tr(\theta)=0$.  In that case this identity is an EA equivalence.  If $\Tr(\theta)=1$, $B_\theta$ is singular and the section-product criterion gives a bad derivative.  The seven admissible nonzero coefficients are therefore the nonzero elements of the trace-zero hyperplane.
\end{proof}

\subsection{Dimension six}

\begin{theorem}\label{thm:n6}
For $n=6$ and $\theta\ne0$, the following are equivalent:
\begin{enumerate}
\item $F_{6,\theta}$ is APN;
\item $\theta^6+\theta^3+1=0$;
\item $\theta^3\in\F_4\setminus\F_2$;
\item $\ord(\theta)=9$.
\end{enumerate}
Thus exactly six coefficients are admissible, and they form one Frobenius orbit.
\end{theorem}

\begin{proof}
Let $E=\F_4=\F_2(\lambda)$ and choose $u$ with $u^3=\lambda$.  Then $K=E(u)=\F_{64}$ and
\[
 T(a+bu+cu^2)=a.
\]
Every pair in \eqref{eq:Sn} has the form
\[
 c=1+pu+qu^2,\qquad d=ru+su^2,
 \qquad p,q,r,s\in E.
\]
The $4^4=256$ products $c^3+d^3$ cover exactly 57 of the 63 nonzero elements.  The complement is
\[
 E^*u\ \cup\ E^*u^2.
\]
Every element of this complement has order nine, and conversely every element of order nine lies in it.  The 12-row Frobenius-orbit certificate in the Supplementary Material gives one exact witness for each represented orbit, so the finite coverage statement is independently checkable.  The equivalences of the four conditions are immediate from $u^3=\lambda$.
\end{proof}

\begin{remark}\label{rem:arshad6}
Example~5.14 of \cite{Arshad2018} lists the same six nonzero coefficients as $\beta^{7j}$ for $j\in\set{1,2,4,5,7,8}$, where $\beta$ has order 63.  Theorem~\ref{thm:n6} supplies the basis-independent order-nine characterization and a completeness proof.
\end{remark}

\subsection{Dimension eight}

\begin{theorem}\label{thm:n8}
For $n=8$ and $\theta\ne0$,
\[
 F_{8,\theta}\text{ is APN}
 \quad\Longleftrightarrow\quad
 \theta\in\F_4^*.
\]
\end{theorem}

\begin{proof}
By Theorem~\ref{thm:section}, it suffices to determine $S_8$.  Write a transverse section as $\phi(y)=cy+dy^2$, with $T(c)=1$ and $T(d)=0$.  Eliminating $d$ from
\[
 T(d)=0,
 \qquad (c^3+d^3)^4+(c^3+d^3)=0
\]
gives a resultant in $c$ whose greatest common divisor with $T(c)+1$ is one.  A complete B\'ezout identity is included in Online Resource~1.  Hence no transverse section product belongs to $\F_4$, proving that every $\theta\in\F_4^*$ is admissible.

For the converse, the $64\cdot64=4096$ pairs with $T(c)=1$ and $T(d)=0$ represent every element of $K\setminus\F_4$.  Frobenius invariance compresses this assertion to the 33 orbit witnesses in the Supplementary Material.  Thus $S_8=K\setminus\F_4$.
\end{proof}

\begin{remark}\label{rem:arshad8}
Example~5.24 of \cite{Arshad2018} gives $\set{1,\beta^{85},\beta^{170}}$ in $\F_{256}^*$.  Since $\beta$ has order 255, this is precisely $\F_4^*$.  Again, the prior computation is acknowledged; Theorem~\ref{thm:n8} gives the intrinsic interpretation and proof.
\end{remark}

The small dimensions can now be summarized as follows.
\begin{table}[ht]
\caption{Admissible nonzero coefficients in the scalar trace-product family}\label{tab:ladder}
\begin{tabular*}{\textwidth}{@{\extracolsep\fill}cll}
\toprule
$n$ & coefficient set & structural interpretation\\
\midrule
4 & $\ker\Tr_{\F_{16}/\F_2}\setminus\set0$ & EA-equivalent to Gold\\
6 & roots of $X^6+X^3+1$ & elements of order 9\\
8 & $\F_4^*$ & relative-trace coefficient fibre\\
even $n\ge10$ & empty & Theorem~\ref{thm:uniform}\\
\botrule
\end{tabular*}
\end{table}

\section{Uniform nonexistence above dimension eight}\label{sec:uniform}

For $\theta\ne0$, define
\begin{equation}\label{eq:hn}
 h_n(\theta)=\#\set{(c,d)\in K^2:T(c)=1,\ T(d)=0,\ c^3+d^3=\theta}.
\end{equation}
By Theorem~\ref{thm:section}, $F_{n,\theta}$ is APN if and only if $h_n(\theta)=0$.

\begin{lemma}[Character-sum bounds for the phase functions]\label{lem:charsum}
Let $C_\theta:X^3+Y^3=\theta Z^3$ and $f_{u,v}=(uX+vY)/Z$ with $(u,v)\ne(0,0)$. If exactly one of $u,v$ is nonzero, the affine additive-character sum is bounded by $6\sqrt q$. If $uv\ne0$, it is bounded by $4\sqrt q+1$.
\end{lemma}
\begin{proof}
The curve is smooth of genus one. In the axis case $f_{u,v}$ has three simple poles, so the Artin--Schreier conductor divisor has degree six and the standard Weil bound gives $6\sqrt q$. In the off-axis case the line $uX+vY=0$ is the flex tangent at one of the three points at infinity; the corresponding apparent pole cancels and two simple poles remain, giving conductor degree four and a complete-sum bound $4\sqrt q$. A function of the form $g^2+g+c$ cannot have a simple pole, so the associated sheaf is nontrivial in both cases. The cancelled point is regular and contributes $1$ to the complete projective sum, whereas the sum used below is affine; this yields $4\sqrt q+1$. See \cite[Thm.~5.2.3]{Stichtenoth2009} and the additive-character formulation in \cite{Lachaud1992}.
\end{proof}

\begin{theorem}[Dimension rigidity]\label{thm:uniform}
For every even $n\ge10$ and every $\theta\ne0$, one has $h_n(\theta)>0$.  Consequently, $F_{n,\theta}$ is not APN.
\end{theorem}

\begin{proof}
Let $q=2^n$ and $\psi(z)=(-1)^{\Tr_{K/\F_2}(z)}$.  Additive-character orthogonality on $E$ gives
\begin{equation}\label{eq:char-expansion}
 h_n(\theta)=\frac1{16}\sum_{u,v\in E}\trchar(u)S_{u,v}(\theta),
 \qquad
 S_{u,v}(\theta)=\sum_{c^3+d^3=\theta}\psi(uc+vd),
\end{equation}
where $\trchar(u)=(-1)^{\tr(u)}$.

Consider the projective Fermat cubic
\[
 C_\theta:X^3+Y^3=\theta Z^3.
\]
It is smooth of genus one.  Since $E\subset K$, the line at infinity contains the three rational points
\[
 P_t=[t:1:0],\qquad t\in E^*.
\]
The affine point count is therefore
\[
 N^{\mathrm{aff}}_\theta=q-2+e_\theta,
 \qquad \abs{e_\theta}\le2\sqrt q.
\]

For a nonzero phase put $f_{u,v}=uX/Z+vY/Z$.  If exactly one of $u,v$ is nonzero, $f_{u,v}$ has three simple poles, one at each $P_t$.  If $u,v\ne0$, the numerator vanishes at $P_{v/u}$.  The line $uX+vY=0$ is the flex tangent there: substituting $X=(v/u)Y$ into the cubic gives $\theta Z^3=0$.  Hence the numerator has intersection multiplicity three, the apparent pole cancels, and only the other two simple poles remain.

Lemma~\ref{lem:charsum} therefore gives
\[
 \abs{S_{u,v}(\theta)}\le6\sqrt q
 \quad\text{for the six axis pairs},\qquad
 \abs{S_{u,v}(\theta)}\le4\sqrt q+1
 \quad\text{for the nine off-axis pairs}.
\]
The zero phase contributes the affine point count.  Thus
\begin{align*}
 16h_n(\theta)
 &\ge q-2-2\sqrt q-6\cdot6\sqrt q-9(4\sqrt q+1)\\
 &=q-11-74\sqrt q.
\end{align*}
This is positive for every even $n\ge14$.

The two remaining dimensions are exact finite bridge cases.  XOR-convolution of the two trace fibres gives
\[
 \min_{\theta\ne0}h_{10}(\theta)=46,
 \qquad
 \min_{\theta\ne0}h_{12}(\theta)=208.
\]
The full histograms and hashes are included in Online Resource~1.  Hence $h_n(\theta)>0$ for every even $n\ge10$.
\end{proof}

\section{Projective rank-two extensions in dimension eight}\label{sec:projective}

Throughout this section use the tower
\[
 \F_2\subset E=\F_4\subset L=\F_{16}\subset K=\F_{256}.
\]
For $c\in L^*$ define
\begin{equation}\label{eq:Qc}
 Q_c(x)=\tr(T(cx))\,\tr(\lambda T(cx)).
\end{equation}
If $c$ is multiplied by an element of $E^*$, the new form differs from $Q_c$ by a linear Boolean function.  Hence its polar class depends only on the projective point $[c]\in\mathbb P_E(L)\cong\mathbb P^1(E)$.

Fix the base point $[1]$.  For a non-base point $p=[\rho]$, define
\begin{equation}\label{eq:sdelta}
 s(p)=\rho^6,
 \qquad
 \delta(p)=\Tr_{L/E}(s(p))^2.
\end{equation}
These quantities are independent of the representative: $\alpha^3=1$ for $\alpha\in E^*$.  Moreover $s(p)$ is the unique element of order five in the projective class.

For $\eta,\zeta\in E^*$ put
\begin{equation}\label{eq:ranktwo}
 G_{p,\eta,\zeta}(x)
 =x^3+\eta Q_1(x)+\zeta\rho^{-3}Q_\rho(x).
\end{equation}
The normalization $\rho^{-3}$ makes both rank-one margins individually APN by Theorem~\ref{thm:n8}.

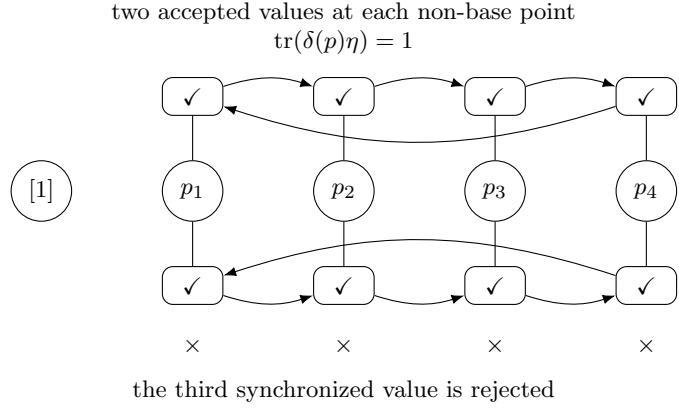
\begin{figure}[ht]
\centering
\begin{tikzpicture}[>=Latex, every node/.style={font=\small}, p/.style={circle,draw,minimum size=8mm}, acc/.style={rounded corners,draw,minimum width=8mm,minimum height=5mm}]
\node[p] (b) at (0,0) {$[1]$};
\foreach \x/\lab in {2/$p_1$,4/$p_2$,6/$p_3$,8/$p_4$}{\node[p] (p\x) at (\x,0) {\lab};}
\foreach \x in {2,4,6,8}{
  \node[acc] (a\x) at (\x,1.25) {$\checkmark$};
  \node[acc] (c\x) at (\x,-1.25) {$\checkmark$};
  \node at (\x,-2.05) {$\times$};
  \draw[-] (p\x) -- (a\x);
  \draw[-] (p\x) -- (c\x);
}
\draw[->,bend left=18] (a2) to (a4);
\draw[->,bend left=18] (a4) to (a6);
\draw[->,bend left=18] (a6) to (a8);
\draw[->,bend left=18] (a8) to (a2);
\draw[->,bend right=18] (c2) to (c4);
\draw[->,bend right=18] (c4) to (c6);
\draw[->,bend right=18] (c6) to (c8);
\draw[->,bend right=18] (c8) to (c2);
\node[align=center] at (4,2.15) {two accepted values at each non-base point\\$\operatorname{tr}(\delta(p)\eta)=1$};
\node[align=center] at (4,-2.65) {the third synchronized value is rejected};
\end{tikzpicture}
\caption{Projective organization in dimension eight. Each of the four non-base points of $\mathbb P^1(\F_4)$ supports two accepted coefficient values and one rejected value. Frobenius organizes the eight accepted marked points into two four-cycles, which become the two EA/CCZ classes of Section~\ref{sec:partition}.}
\label{fig:projective-selector}
\end{figure}

\begin{theorem}[Projective selector]\label{thm:projective}
For every $p=[\rho]\ne[1]$ and $\eta,\zeta\in E^*$,
\begin{equation}\label{eq:projective-selector}
 G_{p,\eta,\zeta}\text{ is APN}
 \quad\Longleftrightarrow\quad
 \eta=\zeta
 \ \text{and}\ 
 \tr\bigl(\delta(p)\eta\bigr)=1.
\end{equation}
Consequently each of the four non-base points supports exactly two normalized APN lifts, for a total of eight marked switchings.
\end{theorem}

\begin{proof}
It is enough to prove the assertion for one order-five representative $r\in L^*$ and transport by $E^*$-scaling and Frobenius.  A transverse derivative section has an inverse
\[
 \phi(y)=ay+by^2,
 \qquad T(a)=1,\quad T(b)=0.
\]
The second trace coordinate induces a transition
\[
 \mathcal L_W(y)=T(r\phi(y))\in\mathrm{GL}(2,2).
\]
Every element of $\mathrm{GL}(2,2)$ has exactly one of the forms
\[
 y\mapsto\alpha y,
 \qquad
 y\mapsto\alpha y^2,
 \qquad \alpha\in E^*.
\]
The first three maps are the identity and the two elements of order three.  For them the trace equations reduce to an Artin--Schreier equation
\[
 z^2+z+1+r=0
\]
in $L$, which has no solution because $\Tr_{L/\F_2}(1+r)=1$.  Hence linear transitions produce no obstruction.

The remaining three maps are semilinear involutions.  Write $K=L(w)$ with $w^2+w=r$, choose $r^5=1$ with $r\ne1$, and put $\lambda=r^2+r^3$.  Their equations reduce to 16 pairs $(u,v)\in L^2$.  After dividing the section product by the common factor $r^2$, the complete multiset is
\begin{center}
\begin{tabular}{@{}lc@{}}
\toprule
normalized product & multiplicity\\
\midrule
$r$ & 2\\
$r^2$ & 2\\
$r+r^2$ & 4\\
$r+r^3$ & 2\\
$1+r^2$ & 2\\
$r+r^2+r^3$ & 2\\
$1+r+r^2$ & 2\\
\bottomrule
\end{tabular}
\end{center}
The three diagonal targets, for $\eta=1,\lambda,\lambda^2$, are respectively $1+r^3$, $r^4$, and $r+r^2$.  Only the last target occurs, with multiplicity four.  The reduction is independent of $\alpha$ because $\alpha^3=1$.  Hence the rejected diagonal pair has $3\cdot4=12$ bad sections, one four-set for each semilinear involution.  A direct reduction of the off-diagonal equations gives six bad sections, while the other two diagonal pairs have none.  The accepted diagonal coefficients are exactly those satisfying the trace condition in \eqref{eq:projective-selector}.

Multiplication of $\rho$ by $E^*$ changes only a linear term, and the four order-five representatives form one Frobenius orbit.  This transports the base calculation to all four non-base points and gives \eqref{eq:sdelta}.
\end{proof}

\begin{table}[ht]
\caption{Bad-section counts for a fixed non-base projective displacement}\label{tab:badcounts}
\begin{tabular*}{\textwidth}{@{\extracolsep\fill}lc}
\toprule
coefficient pair $(\eta,\zeta)$ & number of bad sections\\
\midrule
two diagonal pairs selected by $\tr(\delta(p)\eta)=1$ & 0\\
remaining diagonal pair & 12\\
off-diagonal pair $\eta\ne\zeta$ & 6\\
\botrule
\end{tabular*}
\end{table}

\section{EA and CCZ partition of the marked points}\label{sec:partition}

The eight marked parameter points split into two Frobenius orbits of length four.  The following identity proves that each orbit is an EA class.

\begin{proposition}\label{prop:EAcycle}
Let $G_{\rho,\eta}=G_{[\rho],\eta,\eta}$.  Then
\begin{equation}\label{eq:EAwitness}
 G_{\rho^2,\eta^2}(x)
 =G_{\rho,\eta}(x^{2^7})^2+L_{\rho,\eta}(x),
\end{equation}
where
\begin{equation}\label{eq:linearcorrection}
 L_{\rho,\eta}(x)
 =\eta^2\Tr_{K/\F_2}(x)
 +(\eta\rho^{-3})^2\Tr_{K/\F_2}(\rho^2x)
\end{equation}
is a binary linear map of rank two.
\end{proposition}

\begin{proof}
Apply the inverse binary Frobenius $x\mapsto x^{2^7}$ on the input and Frobenius squaring on the output.  The cube term is fixed.  Squaring the two trace coordinates transports $(\rho,\eta)$ to $(\rho^2,\eta^2)$; the difference between the chosen quadratic representatives is exactly \eqref{eq:linearcorrection}.  The identity was also verified on all 256 inputs for every arrow of the two cycles.
\end{proof}

\begin{theorem}\label{thm:partition}
The eight marked APN switchings of Theorem~\ref{thm:projective} form exactly two EA classes of four points each and exactly two CCZ classes.
\end{theorem}

\begin{proof}
Proposition~\ref{prop:EAcycle} proves equivalence within each Frobenius orbit.  The differential spectra of the two ortho-derivatives are, respectively,
\[
 \set{0^{38184},2^{22179},4^{4338},6^{531},8^{48}}
\]
and
\[
 \set{0^{38256},2^{22116},4^{4230},6^{648},8^{30}}.
\]
They differ, so the two orbits are not EA-equivalent.  The maps are quadratic APN, so CCZ equivalence would imply EA equivalence by \cite{Yoshiara2012}.
\end{proof}

\begin{remark}
Theorem~\ref{thm:partition} does not assert new global APN classes.  It organizes known classes as centre-relative, projectively parametrized switchings.  This distinction is essential: the marked point records how a class is reached from the Gold centre, information forgotten by ordinary EA or CCZ classification.
\end{remark}

\section{Computer assistance and reproducibility}\label{sec:repro}

The discovery computations are not part of the logical proof.  The final dependencies are:
\begin{itemize}
\item Theorem~\ref{thm:n6}: a 12-row Frobenius-orbit certificate representing all 57 forbidden nonzero coefficients in $\F_{64}$.
\item Theorem~\ref{thm:n8}: a symbolic B\'ezout certificate for sufficiency and a 33-row orbit certificate covering $\F_{256}\setminus\F_4$ for necessity.
\item Theorem~\ref{thm:uniform}: an analytic proof for even $n\ge14$ and exact XOR-convolution bridges for $n=10,12$.
\item Theorem~\ref{thm:projective}: an analytic reduction plus a 16-pair table over $\F_{16}$; an independent exhaustive check covers all 108 parameter cases.
\item Theorem~\ref{thm:partition}: explicit EA identities on all field elements and independently computed ortho-derivative spectra.
\end{itemize}

Online Resource~1 contains the Python scripts, canonical JSON outputs, and SHA-256 checksums. The same code and exact certificates are preserved in the archived companion release \cite{KuznetsovTraceSwitchings2026}. The reference implementation uses only exact integer arithmetic and explicit finite-field multiplication; no probabilistic step is used.

\section{Conclusion}\label{sec:conclusion}

A fixed trace-product switching of the Gold cube exhibits a rigid dimension ladder.  In dimension four it is only an output-linear reparametrization of Gold.  In dimension six it selects the six elements of order nine.  In dimension eight it selects $\F_4^*$ and admits a projective rank-two extension governed by $\mathbb P^1(\F_4)$.  In every even dimension at least ten the scalar family disappears completely.

The low-rank criterion itself is centre-independent.  It can therefore be applied to other quadratic APN centres, including quadratic representatives of non-Gold CCZ classes, by replacing the explicit Gold derivative with a derivative-image atlas.  This suggests a broader transition graph whose vertices are EA or CCZ classes and whose edges are low-rank APN updates.  The Gold analysis here supplies a fully resolved local model for that programme.

\section*{Acknowledgements}
The author thanks the developers and maintainers of the open-source mathematical software used for exact verification. OpenAI ChatGPT was used to assist with computational workflow organization, execution of reproducibility checks, LaTeX preparation, and language drafting. The author independently verified all mathematical statements, code outputs, references, and the final text and takes full responsibility for the manuscript.

\section*{Statements and Declarations}

\textbf{Funding.} The author did not receive support from any organization for the submitted work.

\textbf{Competing interests.} The author has no relevant financial or non-financial interests to disclose.

\textbf{Data availability.} All exact finite data used in the proofs are included in Online Resource~1 and are preserved in the Zenodo archive \href{https://doi.org/10.5281/zenodo.21797356}{doi:10.5281/zenodo.21797356}.

\textbf{Code availability.} Reproducible Python scripts and canonical JSON certificates are available in the companion GitHub repository \url{https://github.com/KuznetsovKarazin/apn-trace-product-switchings}. The immutable archival release supporting this article is available at \href{https://doi.org/10.5281/zenodo.21797356}{doi:10.5281/zenodo.21797356}.

\textbf{Author contributions.} Oleksandr Kuznetsov conceived the study, developed the mathematical analysis, implemented and validated the computations, and wrote the manuscript.

\textbf{Author identifier.} ORCID: 0000-0003-2331-6326.

\textbf{Ethics approval, consent to participate, and consent for publication.} Not applicable.

\bibliography{references}
\end{document}